\pdfoutput=1 
\documentclass[a4paper,UKenglish,cleveref, autoref, thm-restate]{lipics-v2021}

\hideLIPIcs  

\title{Strictification of weakly stable type-theoretic structures using generic contexts}

\author{Rafaël Bocquet}{Eötvös Loránd University, Budapest, Hungary}{bocquet@inf.elte.hu}{https://orcid.org/0000-0001-6484-9570}
{}

\authorrunning{R. Bocquet} 

\Copyright{Rafaël Bocquet} 



\category{} 

\relatedversion{} 

\nolinenumbers 

\EventEditors{Henning Basold, Jesper Cockx, and Silvia Ghilezan}
\EventNoEds{3}
\EventLongTitle{27th International Conference on Types for Proofs and Programs (TYPES 2021)}
\EventShortTitle{TYPES 2021}
\EventAcronym{TYPES}
\EventYear{2021}
\EventDate{June 14--18, 2021}
\EventLocation{Leiden, The Netherlands (Virtual Conference)}
\EventLogo{}
\SeriesVolume{239}
\ArticleNo{2}

\newtheorem{construction}[theorem]{Construction}

\usepackage[usenames,dvipsnames]{xcolor}

\usepackage{tikz}
\usetikzlibrary{cd}
\usetikzlibrary{decorations.pathmorphing}

\usepackage{stmaryrd}
\usepackage{amsmath}
\usepackage{amsthm}
\usepackage{amssymb}
\usepackage{mathtools}
\usepackage{mathpartir}
\usepackage{environ}
\usepackage{enumerate}
\usepackage{bm}
\usepackage{bbm}
\usepackage{hyperref}
\usepackage{xspace}
\usepackage{mathbbol}

\input{macros.tex}

\newcommand{\inner}[1]{{\color{RedOrange}{\mathbf{#1}}}}

\newcommand{\iS}{\inner{S}}

\newcommand{\iX}{\inner{X}}
\newcommand{\iY}{\inner{Y}}
\newcommand{\opf}{\inner{f}}
\newcommand{\opg}{\inner{g}}

\newcommand{\bty}{\mathbb{ty}}
\newcommand{\btm}{\mathbb{tm}}

\newcommand{\BSort}{\mathsf{BSort}}
\newcommand{\MonoSort}{\mathsf{MonoSort}}
\newcommand{\PolySort}{\mathsf{PolySort}}
\newcommand{\Elem}{\mathsf{Elem}}

\newcommand{\Free}{\mathsf{Free}}

\newcommand{\GenTy}{\mathsf{GenTy}}
\newcommand{\GenTm}{\mathsf{GenTm}}

\newcommand{\Nf}{\mathsf{Nf}}
\newcommand{\NfTy}{\mathsf{NfTy}}
\newcommand{\nf}{\mathsf{nf}}

\begin{document}

\maketitle

\begin{abstract}
  We present a new strictification method for type-theoretic structures that are only weakly stable under substitution.
  Given weakly stable structures over some model of type theory, we construct equivalent strictly stable structures by evaluating the weakly stable structures at generic contexts.
  These generic contexts are specified using the categorical notion of familial representability.
  This generalizes the local universes method of Lumsdaine and Warren.

  We show that generic contexts can also be constructed in any category with families which is freely generated by collections of types and terms, without any definitional equality.
  This relies on the fact that they support first-order unification.
  These free models can only be equipped with weak type-theoretic structures, whose computation rules are given by typal equalities.
  Our main result is that any model of type theory with weakly stable weak type-theoretic structures admits an equivalent model with strictly stable weak type-theoretic structures.
\end{abstract}

\section{Introduction}

Type-theoretic structures are usually required to be strictly stable under substitution.
However many structures arising from category theory and homotopy theory are only specified up to isomorphism, equivalence or homotopy.
They are then only \emph{weakly} stable under substitution.
This is for instance the case for the identity types arising from weak factorization systems~\cite{HomotopyTheoreticModelsOfIdTypes} and for the constructive simplicial model of Gambino and Henry~\cite{TowardsConstructiveSimplicialModel}.
In order to interpret type theories into such structures, we have to use \emph{strictification theorems} that replace weakly stable structures by strictly stable ones.

Generally, a strictification method is a procedure that constructs, given an input model with weakly stable type structures, another model with stable type structures, connected to the original model via a zigzag of equivalences (for a suitable notion of equivalence).
Several strictification methods are known~\cite{HofmannInterpretationTTInLCCCs,RevisitingCategoricalInterpretationDTT,LocalUniversesModel,CurienSubstUpToIso,InterpretationDTTInModelCategoryLCCCs}, with different constraints on the type theories and models.
We recall two of the most general constructions.
\begin{description}
\item[Right adjoint splitting:] A strictification method~\cite{HofmannInterpretationTTInLCCCs,RevisitingCategoricalInterpretationDTT} due to Hofmann defines a new model $\CC_{\star}$ in which types over a context $\Gamma$ are coherent families of types of the base model $\CC$, indexed by the substitutions $\Delta \to \Gamma$.
  This is a \emph{cofree} construction: we pack together all the data that is needed when substituting, along with witnesses that this data is coherent, \ie{} that different ways of substituting coincide, up to isomorphism or equivalence.

  This method is known to work for extensional type theories, \ie{} type theories with the equality reflection rule, but it does not directly apply to most models arising from homotopy theory.
  In presence of equality reflection it is sufficient to consider families of types that are coherent up to isomorphism.
  A generalization would need to consider homotopy-coherent families of types and terms, that include coherence conditions in all dimensions.
  Defining a workable notion of homotopy-coherent family is however not easy.

  We note that coherence theorems proven in Uemura's PhD thesis~\cite{UemuraThesis} essentially involve such homotopy-coherent families.

\item[Left adjoint splitting/local universes:] The local universes method~\cite{LocalUniversesModel} of Lumsdaine and Warren generalizes Voevodsky's use of universes to obtain stability in the simplicial model~\cite{SimplicialModelUF}.
  It instantiates the weakly stable structures at suitable \emph{generic contexts}.
  Strict stability under substitution then follows from the stability of the construction of the generic contexts.
  In order to ensure the existence of the generic contexts, this strictification method replaces the base model $\CC$ by a new model, the \emph{local universes model} $\CC_{!}$, also called the \emph{left adjoint splitting}, in which types over $\Gamma$ are replaced by triples $(V,E,\chi)$, where $(V,E)$ is a local universe, consisting of a closed context $V$ and of a type $E$ over $V$, and $\chi$ is a substitution from $\Gamma$ to $V$.
  The generic contexts of the type and term formers then only depend on the local universes of the type parameters, but not on the map $\chi$ nor on the term parameters; this ensures that they are invariant under substitution.
  The construction of these generic contexts requires the existence of some local exponentials in the underlying category of the base model.
  This condition is called condition~(LF).
\end{description}

\paragraph*{Generic contexts}
We present a new general strictification method.
Like the local universes method, our method instantiates the weakly stable structures at generic contexts.
In the local universes construction, the generic contexts can only depend on the shapes of types, but not on the structure of terms.
We give a finer characterization of the (universal) properties required by the generic contexts, using the categorical notion of \emph{familial representability}~\cite{ConnectedLimitsFamilialRepresentabilityArtinGlueing,CorrigendaConnectedLimitsFamilialRepresentabilityArtinGlueing}.

If $x$ is an element over a context $\Gamma$ of a presheaf $X$ (such as the presheaf of types or the presheaf of terms of a given type), a generalization of $x$ is an element $x_{0}$ over some context $\Gamma_{0}$, along with a substitution $\rho : \Gamma \to \Gamma_{0}$ such that $x = x_{0}[\rho]$.
A most general generalization is a terminal object in the category of generalizations.
When they exist, the most general generalizations of $x$ and $x[\sigma]$ coincide (at least up to isomorphism).
The presheaf $X$ is \emph{familially representable} if all of its elements admit most general generalizations (with some additional naturality condition).
Equivalently, a presheaf is familially representable when it is a coproduct of a family of representable presheaves.

A weakly stable type-theoretic operation (type or term former) $T$ on a category $\CC$ is given by a dependent non-natural transformation
$T : \forall (\Gamma : \Ob_{\CC}) (x : X_{\Gamma}) \to Y_{\Gamma}(x)$,
where $X$ is a presheaf over $\CC$ and $Y$ is a dependent presheaf over $X$.
When the presheaf $X$ is familially representable, we can define a natural transformation
$T^{s} : \forall (\Gamma : \CC^{\op}) (x : X_{\Gamma}) \to Y_{\Gamma}(x)$ by
$T^{s}(\Gamma,x) \triangleq T(\Gamma_{0},x_{0})[\rho]$,
where $x_{0} : X_{\Gamma_{0}}$ is the most general generalization of $x$.
Here we have defined a strictly stable operation $T^{s}$ as the instantiation of the weakly stable operation $T$ at the \emph{generic context} $\Gamma_{0}$.

The presheaves $X$ that may occur as the sources of type-theoretic operations all have a specific shape: they are given by \emph{polynomial sorts}, which are obtained by closing the basic sorts (types and terms) under dependent products (with arities in terms) and dependent sums.
They correspond to the objects of the representable map category~\cite{UemuraFramework} that encodes the type theory.
We say that a model (a category with families) has \defemph{familially representable polynomial sorts} when the presheaves of elements of polynomial sorts are all familially representable.
Any weakly stable type-theoretic structure over a base model that satisfies that condition can be replaced by a stable type-theoretic structure.

We obtain the following theorem.
\begin{restatable}{theorem}{restateStrictId}\label{thm:lfr_poly_strict_id}
  Let $\CC$ be a CwF equipped with weakly stable identity types.
  If $\CC$ has familially representable polynomial sorts, then $\CC$ can be equipped with stable identity types that are equivalent to the weakly stable identity types.
\end{restatable}
It is straightforward to generalize this construction to any other weakly stable type-theoretic structure.

The condition~(LF) of the local universe method~\cite{LocalUniversesModel} implies that the local universe model $\CC_{!}$ has familially representable polynomial sorts; thus the local universe method factors through our method.

\paragraph*{Free categories with families}

There are models that have familially representable polynomials sorts without satisfying condition~(LF).
We show that this is the case for all categories with families (CwFs) that are freely generated by some collection of generating types and terms.
Freely generated CwFs can also be seen as generalized (\ie{} dependently sorted) algebraic theories~\cite{CartmellGATs} without equations.
Using the terminology of weak factorization systems, the freely generated CwFs can be described as the cellular objects with respect to some set $I$ of CwF morphisms.

Thanks to the absence of equations, free CwFs support \emph{first-order unification}; any two unifiable types, terms or substitutions admit a \emph{most general unifier}.
These most general unifiers are used to construct most general generalization for polynomial sorts.
\begin{restatable}{theorem}{restateMggFree}\label{thm:mgg_free}
  If a CwF $\CC$ is freely generated ($I$-cellular), then it has locally familially representable polynomials sorts.
\end{restatable}

\paragraph*{Strictification of weakly stable weak type-theoretic structures}

By the small object argument, every CwF $\CC$ admits an $I$-cellular replacement, which is a freely generated CwF $\CC_{0}$ equipped with a \emph{trivial fibration} $F : \CC_{0} \tfibra \CC$.
A trivial fibration is a morphism that is surjective on types and terms; in particular it is a kind of equivalence between CwFs.
Thus every CwF $\CC$ admits an equivalent CwF $\CC_{0}$ that has familially representable polynomial sorts.
Furthermore all type and term formers can be lifted from $\CC$ to $\CC_{0}$ along $F$, except that definitional equalities cannot be lifted.

In other words, every \emph{weak} type-theoretic structure can be lifted.
A weak type-theoretic structure is a type-theoretic structure that is presented without definitional equalities.
Typically, their computation rules are specified up to typal equality, rather than up to definitional equality.
For example, weak identity types (under the name of propositional identity types) were introduced in~\cite{PropositionalIdentityTypes}.
The computation rule of the weak $\J$ eliminator is only given by a typal equality ${\Jbeta} : \Id(\J(d,x,\refl),d)$.
Similarly, we can consider weak $\Pi$-types, weak $\Sigma$-types, \etc.

We thus have two ways to weaken the usual presentation of a type-theoretic structure: we can weaken either the stability under substitution and/or the computation rules.
In general we may want to compare weakly stable, weakly computational structures with strictly stable, strictly computational structures.
As it is hard to do this comparison directly, it has to be split into multiple steps.
The present paper provides comparisons between weakly stable, weak and strictly stable, weak structures.
There is ongoing work~\cite{CoherenceStrictEqualities} by the author towards coherence theorems that compare strictly stable, weak structures with strictly stable, strict structures.

Combining the previous results, we obtain the following theorem:
\begin{restatable}{theorem}{restateLeftStrictification}\label{thm:left_strictification_id}
  Let $\CC$ be a CwF with weakly stable weak identity types.
  Then there exists a CwF $\CD$ with stable weak identity types and a trivial fibration $F : \CD \to \CC$ in $\CCwf$ that weakly preserves identity types.
\end{restatable}
This theorem can straightforwardly be extended to any other weakly stable weak type-theoretic structure.

In general, we are interested in coherence theorems that are more powerful than~\cref{thm:left_strictification_id}.
We expect that~\cref{thm:left_strictification_id} can be part of the proofs of such coherence theorems; this is discussed in~\cref{sec:towards_coherence}.

\section{Background}

We work in a constructive metatheory.
\subsection{Presheaf categories}

We use the internal language of the category $\CPsh(\CC)$ of presheaves over a base category $\CC$; any presheaf category is a model of extensional type theory~\cite{SyntaxAndSemanticsDT}.
This justifies the use of higher-order abstract syntax (HOAS) to describe type-theoretic structures over a base category $\CC$.

If $\Gamma : \Ob_{\CC}$ is an object of $\CC$, the corresponding representable presheaf is written $\yo(\Gamma)$.
A morphism $f : \Gamma \to \Delta$ can be identified with the natural transformation $f : \yo(\Gamma) \to \yo(\Delta)$.

If $X$ is a presheaf, we identify global elements of the exponential presheaf $(\yo(\Gamma) \to X)$ with elements of the evaluation of $X$ at $\Gamma$.
If $x : \yo(\Gamma) \to X$ and $f : \Delta \to \Gamma$, we may write $x[f]$ for the restriction of $x$ along $f$.

We write $\int_{\CC} X$ for the category of elements of $X$; its objects are pairs $(\Gamma,x)$ with $x : \yo(\Gamma) \to X$, and a morphism $(\Delta,x') \to (\Gamma,x)$ is a morphism $\rho : \Delta \to \Gamma$ such that $x' = x[\rho]$.

A dependent presheaf over $X$ is a presheaf over $\int_{\CC} X$.
If $Y$ is a dependent presheaf over $X$ and $x : \yo(\Gamma) \to X$, global elements of the presheaf $(\gamma : \yo(\Gamma)) \to Y(x(\gamma))$ coincide with elements of the evaluation of $Y$ at $\Gamma$ and $x$.

The presheaf universe classifying the $i$-small dependent presheaves is denoted by $\UU_{i}$; we will generally omit the universe level $i$.
Dependent products are written $(a : A) \to B(a)$, sometimes with a leading $\forall$ quantifier.
Dependent sums are written $(a : A) \times B(a)$.
The terminal presheaf is denoted by $\top$.

If $x : \yo(\Gamma) \to X$ and $y : (\gamma : \yo(\Gamma)) \to Y(x(\gamma))$, we write $\angles{x,y}$ for the corresponding element of $(\gamma : \yo(\Gamma)) \to (a : X(\gamma)) \times (b : Y(a))$.
We write $\angles{}$ for the unique element of $\yo(\Gamma) \to \top$.

\subsection{Categories with Families}\label{ssec:cwfs}

We use categories with families~\cite{InternalTypeTheory,CwFsUSD} as our models of type theory.
We recall how the notion of local representability, which encodes the context extensions, is derived from the (non-local) notion of representability.
We will similarly derive a notion of local familial representability from the notion of familial representability in~\cref{ssec:fam_rep}.

\begin{definition}
  A dependent presheaf $Y : X \to \UU$ is \defemph{locally representable} if for every element $x : \yo(\Gamma) \to X$, the restricted presheaf
  \begin{alignat*}{3}
    & Y_{\mid x} && :{ } && \CPsh(\CC / \Gamma) \\
    & Y_{\mid x}(\rho : \Delta \to \Gamma) && \triangleq{ } && Y(x[\rho] : \yo(\Delta) \to X)
  \end{alignat*}
  is representable.
  \lipicsEnd{}
\end{definition}

\begin{definition}
  A \defemph{family} over a category $\CC$ is a pair $(\Ty,\Tm)$ consisting of a presheaf $\Ty : \UU$ and of a dependent presheaf $\Tm : \Ty \to \UU$.
  We say that the family has \defemph{representable elements} when $\Tm$ is locally representable.
  \lipicsEnd{}
\end{definition}

\begin{definition}
  A \defemph{category with families} (CwF) is a category $\CC$ equipped with a terminal object $\diamond$, along with a global family $(\Ty_{\CC},\Tm_{\CC})$ with representable elements.
  \lipicsEnd{}
\end{definition}

The local representability condition describes the context extensions.
If $\Gamma : \Ob_{\CC}$ and $A : \yo(\Gamma) \to \Ty_{\CC}$, we have an extended context $\Gamma.A : \Ob_{\CC}$ and a natural isomorphism $\yo(\Gamma.A) \simeq (\gamma : \yo(\Gamma)) \times (a : \Tm_{\CC}(A(\gamma)))$.
We will often identify the two sides of this isomorphism.
The two projections out of this dependent sum are the projection morphism $\bm{p}_{A} : \Gamma.A \to \Gamma$ and the variable term $\bm{q}_{A} : ((\gamma,a) : \yo(\Gamma.A)) \to \Tm_{\CC}(A(\gamma))$.
If $\rho : \Delta \to \Gamma$, we write $\rho^{+}$ for the canonical morphism $\rho^{+} : \Delta.A[\rho] \to \Gamma.A$, \ie{} $\rho^{+} = \angles{\rho \circ \bm{p}_{A}, \bm{q}_{A}}$.

We write $\CCwf$ for the $1$-category of CwFs and strict CwF morphisms.

We write $(\Ty^{\star}, \Tm^{\star})$ for the family of \defemph{telescopes} of a family $(\Ty,\Tm)$.
It is defined as the following inductive-recursive family, internally to $\CPsh(\CC)$:
\begin{alignat*}{3}
  & \Ty^{\star} && :{ } && \UU \\
  & \Tm^{\star} && :{ } && \Ty^{\star} \to \UU \\
  & \diamond && :{ } && \Ty^{\star} \\
  & \Tm^{\star}(\diamond) && \triangleq{ } && \top \\
  & \_{}. \_{} && :{ } && (\Delta : \Ty^{\star}) (A : \Tm^{\star}(\Delta) \to \Ty) \to \Ty^{\star} \\
  & \Tm^{\star}(\Delta. A) && \triangleq{ } && (\delta : \Tm^{\star}(\Delta)) \times (a : \Tm(A(\delta)))
\end{alignat*}

In other words, a telescope of types $A : \Ty^{\star}$ is a finite sequence $A_{1}.A_{2}.\cdots.A_{n}$ of dependent types.
A telescope of terms $a : \Tm^{\star}(A)$ is a sequence $a_{1} : A_{1}$, $a_{2} : A_{2}(a_{1})$, \dots, $a_{n} : A_{n}(a_{1},a_{2},\dotsc)$ of terms.
If $(\Ty,\Tm)$ has representable elements, then so does $(\Ty^{\star},\Tm^{\star})$; the context extensions of $(\Ty^{\star},\Tm^{\star})$ are iterations of the context extensions of $(\Ty,\Tm)$.

There is a canonical map $\Ty^{\star}_{\CC} \to \Ob_{\CC}$ sending any closed telescope to the corresponding extension of the empty context.
We say that $\CC$ is \defemph{contextual} when that map is bijective.
In that case, we identify the objects of $\CC$ and the closed telescopes.
Up to that identification, the Yoneda embedding $\yo : \Ob_{\CC} \to \UU$ coincides with the restriction of $\Tm^{\star}_{\CC} : \Ty^{\star}_{\CC} \to \UU$ to closed telescopes.

\begin{definition}
  If $\CC$ is a contextual CwF, we characterize its variables by an inductive family $\Var : (\Gamma : \Ob_{\Gamma})(A : \yo(\Gamma) \to \Ty_{\CC})(a : \forall \gamma \to \Tm_{\CC}(A(\gamma))) \to \SSet$, generated by:
  \begin{mathpar}
    \inferrule{{}}{\Var_{\Gamma.A,A[\bm{p}_{A}]}(\bm{q}_A)}

    \inferrule{\Var_{\Gamma,A}(x)}{\Var_{\Gamma.B,A[\bm{p}_{B}]}(x[\bm{p}_{B}])}
  \end{mathpar}
\end{definition}

\subsection{Strictly stable and weakly stable weak identity types}

We give definitions of the structures of stable and weakly stable weak identity types using the internal language of $\CPsh(\CC)$.
Note that the weakly stable structures cannot be fully be specified internally; it involves an external quantification over contexts.

We use Paulin-Mohring's variant of the identity type elimination principle, as it is better behaved than Martin-Löf's eliminator in the absence of other type-theoretic structures.
In the absence of $\Pi$-types, Martin-Löf's eliminator needs to be parametrized by an additional telescope, as introduced by Gambino and Garner~\cite{IdentityTypeWFS}.
This is discussed in more details in~\cite{NorthIdentityTypes,IsaevIndexedTypeTheories,CoherenceStrictEqualities}.

Paulin-Mohring's eliminator corresponds to based path induction, in which the left endpoint of a path is fixed.
\begin{mathpar}
  \inferrule{A\ \type \\ x : A}{[y:A]\ \Id(A,x,y)\ \type}

  \inferrule{A\ \type \\ x : A}{\refl(A,x) : \Id(A,x,x)}

  \inferrule{
    A\ \type \\
    x : A
    \\\\
    [y:A, p:\Id(A,x,y)]\ P(y,p)\ \type \\
    d : P(x,\refl(A,x))
  }
  {[y:A, p:\Id(A,x,y)]\ \J(A,x,P,d,y,p) : P(y,p)}
\end{mathpar}

We consider \emph{weak} identity types, which means that their computation rule is given by a typal equality, rather than a definitional equality.
\begin{mathpar}
  \inferrule{
    A\ \type \\
    x : A
    \\\\
    [y:A, p:\Id(A,x,y)]\ P(y,p)\ \type \\
    d : P(x,\refl(A,x))
  }
  {\Jbeta(A,x,P,d,y,p) : \Id(P(x,\refl(A,x)), \J(A,x,P,d,x,\refl(A,x)), d)}
\end{mathpar}

Note that the type former $\Id$ has two parameters ($A$ and $x$) and one index $y$.
The fact that $y$ is an index cannot be seen in the definition of the stable type-former $\Id$ as a natural transformation $\Id : (A : \Ty_{\CC})(x,y : \Tm_{\CC}(A)) \to \Ty_{\CC}$.
However it changes the definition of the weakly stable type-former $\Id$; we will have a type $\Id_{\Gamma,A,x}$ in the extended context $\Gamma.(y:A)$.

\begin{definition}
  A (strictly stable) \defemph{weak identity type structure} on a family $(\Ty,\Tm)$ consists of an \defemph{introduction structure}
  \begin{alignat*}{3}
    & \Id && :{ } && \forall (A : \Ty) (x,y : \Tm(A)) \to \Ty, \\
    & \refl && :{ } && \forall A\ x \to \Tm(\Id(A,x,x)),
  \end{alignat*}
  along with a \defemph{weak elimination structure}
  \begin{alignat*}{3}
    & \J && :{ } && \forall (A : \Ty) (x : \Tm(A)) \\
    &&&&& \phantom{\forall} (P : \forall (y : \Tm(A)) (p : \Tm(\Id(A,x,y))) \to \Ty) \\
    &&&&& \phantom{\forall} (d : \Tm(P(x,\refl(A,x)))) \\
    &&&&& \to \forall y\ p \to \Tm(P(y,p)), \\
    & \Jbeta && :{ } && \forall (A : \Ty) (x : \Tm(A)) \\
    &&&&& \phantom{\forall} (P : \forall (y : \Tm(A)) (p : \Tm(\Id(A,x,y))) \to \Ty) \\
    &&&&& \phantom{\forall} (d : \Tm(P(x,\refl(A,x)))) \\
    &&&&& \to \Tm(\Id(P(x,\refl(A,x)), \J(A,x,P,d,x,\refl(A,x)), d)).
          \tag*{\lipicsEnd{}}
  \end{alignat*}
\end{definition}

We also define the weakly stable weak identity types.
\begin{definition}
  A \defemph{$\Id$-introduction context} is a triple $(\Gamma,A,x)$, where
  \begin{alignat*}{3}
    & \Gamma && :{ } && \Ob_{\CC}, \\
    & A && :{ } && \yo(\Gamma) \to \Ty, \\
    & x && :{ } && (\gamma : \yo(\Gamma)) \to \Tm(A(\gamma)).
  \end{alignat*}
  Here $\Gamma$ is an object of $\CC$, and $A$ and $x$ are types and terms that only depend on $\Gamma$.

  A \defemph{weakly stable identity type introduction structure} consists, for every $\Id$-introduction context $(\Gamma,A,x)$, of operations
  \begin{alignat*}{3}
    & \Id_{(\Gamma,A,x)} && :{ } && \forall (\gamma : \yo(\Gamma)) (y : \Tm(A(\gamma))) \to \Ty, \\
    & \refl_{(\Gamma,A,x)} && :{ } && \forall (\gamma : \yo(\Gamma)) \to \Tm(\Id_{(\Gamma,A,x)}(\gamma,x(\gamma))).
  \end{alignat*}

  A \defemph{$\Id$-elimination context} over an $\Id$-introduction context $(\Gamma,A,x)$ is a tuple $(\Delta,\gamma,P,d)$, where
  \begin{alignat*}{3}
    & \Delta && :{ } && \Ob_{\CC}, \\
    & \gamma && :{ } && \Delta \to \Gamma, \\
    & P && :{ } && \forall (\delta : \yo(\Delta)) (y : \Tm(A(\gamma(\delta)))) (p : \Tm(\Id_{(\Gamma,A,x)}(\gamma(\delta), y))) \to \Ty, \\
    & d && :{ } && \forall (\delta : \yo(\Delta)) \to \Tm(P(\delta, x(\gamma(\delta)), \refl_{(\Gamma,A,x)}(\gamma(\delta)))).
  \end{alignat*}

  A \defemph{weakly stable identity type elimination structure} consists, for every $\Id$-elimination context $(\Delta,\gamma,P,d)$ over $(\Gamma,A,x)$, of operations
  \begin{alignat*}{3}
    & \J_{(\Gamma,A,x,\Delta,\gamma,P,d)} && :{ }
    && \forall (\delta : \yo(\Delta)) (y : \Tm(A(\gamma(\delta)))) (p : \Tm(\Id_{(\Gamma,A,x)}(\gamma(\delta), y))) \to \Tm(P(\delta,y,p)), \\
    & \Jbeta_{(\Gamma,A,x,\Delta,\gamma,P,d)} && :{ }
    && \forall (\delta : \yo(\Delta)) \to
       \Id_{(\Delta,P',d)}
       (\delta, \J_{(\Gamma,A,x,\Delta,\gamma,P,d)}(\delta, x(\gamma(\delta)), \refl_{(\Gamma,A,x)}(\gamma(\delta)))), \\
    & P'(\delta') && \triangleq{ }
    && P(\delta', x(\gamma(\delta')), \refl_{\Gamma}(\gamma(\delta'))).
       \tag*{\lipicsEnd{}}
  \end{alignat*}
\end{definition}

Note that strictly stable identity types are weakly stable identity types satisfying additional naturality conditions.
In presence of weakly stable weak identity types, we have well-behaved notions of contractible types, type equivalences, \etc{}

\begin{proposition}\label{prop:ws_wsid}
  The weakly stable weak identity types are indeed \emph{weakly stable}: for every $\Id$-introduction context $(\Gamma,A,x)$ and substitution $\rho : \Delta \to \Gamma$, the canonical map
  \begin{alignat*}{3}
    & \Tm(\Id_{(\Delta,A[\rho],x[\rho])}) && \to{ } && \Tm(\Id_{(\Gamma,A,x)}[\rho])
  \end{alignat*}
  is an equivalence over $\Delta.A[\rho]$.
  \qed{}
\end{proposition}

\begin{definition}\label{def:preserves_wsid}
  A CwF morphism $F : \CC \to \CD$ weakly preserves weakly stable weak identity types if for every $\Id$-introduction context $(\Gamma,A,x)$ of $\CC$, then the canonical map
  \[ \Tm_{\CD}(\Id_{(F(\Gamma),F(A),F(x))}) \to \Tm_{\CD}(F(\Id_{(\Gamma,A,x)})) \]
  is an equivalence over $F(\Gamma.A)$.
  \qed{}
\end{definition}

\subsection{Trivial fibrations and freely generated CwFs}

We recall the definition of the (cofibrations, trivial fibrations) weak factorization system on $\CCwf$.
The same weak factorization system on the category $\mathbf{CwA}$ of Categories with Attributes, which is equivalent to $\CCwf$, was introduced by Kapulkin and Lumsdaine~\cite[Definition~{4.12}]{HomotopyTheoryTTs}.

\begin{definition}\label{def:trivial_fibration}
  A morphism $F : \CC \to \CD$ of CwFs is a \defemph{trivial fibration} if its actions on types and terms are surjective, \ie{} if it satisfies the following lifting conditions:
  \begin{description}
  \item[(type lifting)] For every object $\Gamma : \Ob_{\CC}$ and type $A : \yo(F(\Gamma)) \to \Ty_{\CD}$, there exists a type $A_{0} : \yo(\Gamma) \to \Ty_{\CC}$ such that $F(A_{0}) = A$.
  \item[(term lifting)] For every object $\Gamma : \Ob_{\CC}$, type $A : \yo(\Gamma) \to \Ty_{\CC}$ and term $a : (\gamma : \yo(F(\Gamma))) \to \Tm_{\CD}(F(A)(\gamma))$, there exists a term $a_{0} : (\gamma : \yo(\Gamma)) \to \Tm_{\CC}(A(\gamma))$ such that $F(a_{0}) = a$,
  \end{description}
  where the existential quantifications are strong, meaning that $F$ is equipped with a choice of lifts.
  \lipicsEnd{}
\end{definition}

The (cofibrations, trivial fibrations) weak factorization system on $\CCwf$ is cofibrantly generated by the set $I = \{I^{\ty}, I^{\tm}\}$, where
\begin{mathpar}
  I^{\ty} : \Free(\bm{\Gamma} : \Ob) \to \Free(\bm{A} : \yo\bm{\Gamma} \to \Ty),

  I^{\tm} : \Free(\bm{A} : \yo\bm{\Gamma} \to \Ty) \to \Free(\bm{a} : (\gamma : \yo\bm{\Gamma}) \to \Tm(\bm{A}(\gamma))).
\end{mathpar}
Here $\Free(\bm{\Gamma} : \Ob)$ is the CwF freely generated by an object $\bm{\Gamma}$,
$\Free(\bm{A} : \yo\bm{\Gamma} \to \Ty)$ is the CwF freely generated by an object $\bm{\Gamma}$ and a type $\bm{A}$ over $\Gamma$, and
$\Free(\bm{a} : (\gamma : \yo\bm{\Gamma}) \to \Tm(\bm{A}(\gamma)))$ is the CwF freely generated by $\bm{\Gamma}$, $\bm{A}$ and a term $\bm{a}$ of type $\bm{A}$ over $\bm{\Gamma}$.

We also recall the definition of $I$-cellular maps and objects in $\CCwf$.
\begin{definition}\label{def:cellular}
  A \defemph{basic $I$-cellular map} $\CC \to \CD$ is a pushout of a coproducts of maps in $I$; it freely adjoins to a model $\CC$ a collection of new types and terms whose contexts and types are from $\CC$.
  An \defemph{$I$-cellular map} is a sequential composition of a sequence ${(\iota_{i}^{i+1} : \CC_{i} \to \CC_{i+1})}_{i\le\omega}$ of basic $I$-cellular maps.

  A CwF $\CC$ is an \defemph{$I$-cellular object} (or \defemph{$I$-cell complex}) if the unique map $\Init \to \CC$ is an $I$-cellular map.
  \lipicsEnd{}
\end{definition}

By the small object argument, every morphism of CwFs can be factored as an $I$-cellular map followed by a trivial fibration.
In particular, for any CwF $\CC$, the factorization of the unique map $\Init \to \CC$ provides an $I$-cellular object $\CC_{0}$ and a trivial fibration $\CC_{0} \to \CC$.

\begin{proposition}\label{prop:tfib_lift_id}
  If $F : \CC \to \CD$ is a trivial fibration between CwFs and $\CD$ is equipped with weakly stable weak identity types, then $\CC$ can be equipped with weakly stable weak identity types that are strictly preserved by $F$.
\end{proposition}
\begin{proof}
  By lifting each component of the weakly stable weak identity types of $\CD$.
\end{proof}

\begin{proposition}
  Any $I$-cellular CwF is contextual.
\end{proposition}
\begin{proof}
  Let $\Nat$ be the terminal contextual CwFs; its contexts are natural numbers, and it has a unique type and a unique term over every context.
  A CwF $\CC$ is contextual if and only there exists a unique CwF morphism $\CC \to \Nat$; such a morphism gives the length of every context of $\CC$.

  Now take an $I$-cellular CwF $\CC$.
  For any other CwF $\CD$, a CwF morphism $\CC \to \CD$ is determined by the image of the generating types and terms of $\CC$.
  Since $\Nat$ has a unique type and a unique term, there exists a unique CwF morphism $\CC \to \Nat$, sending each generating type or term to the unique type or term of $\Nat$.
  Thus $\CC$ is contextual, as needed.
\end{proof}

The collections of generating types and terms of an $I$-cellular CwF $\CC$ can be obtained from the decomposition of $\Init \to \CC$ as an $I$-cellular map.
We use a ($\inner{red}$,$\inner{bold}$) font to distinguish the generating types and terms from arbitrary types and terms.
\begin{construction}\label{constr:free_gens}
  Let $\CC$ be an $I$-cellular CwF.
  Then we construct sets $\GenTy_{\CC} : \SSet$ of \defemph{generating types} and $\GenTm_{\CC} : \SSet$ of \defemph{generating terms} such that
  \begin{itemize}
  \item For every $\iS : \GenTy_{\CC}$, we have an object $\partial \iS : \Ob_{\CC}$ and a dependent type $\iS : \partial \iS \to \Ty_{\CC}$.
  \item For every $\opf : \GenTm_{\CC}$, we have an object $\partial \opf : \Ob_{\CC}$, a type $T \opf : \partial \opf \to \Ty_{\CC}$
    and a dependent term $\opf : \forall (\tau : \partial \opf) \to \Tm_{\CC}(T \opf(\tau))$.
  \end{itemize}
  The components $\partial \iS$ and $\partial \opf$ specify the \emph{dependencies} (or the \emph{boundary}) of the generating types and terms.
  The component $T \opf$ gives the output type of a generating term.
\end{construction}
\begin{proof}[Construction]
  Since $\CC$ is $I$-cellular, it is the colimit of a sequence
  \[ {(\iota_{i}^{i+1} : \CC_{i} \to \CC_{i+1})}_{i<\omega}
  \] of basic $I$-cellular maps, with $\CC_{0} = \Init_{\CCwf}$ and $\CC_{\omega} = \CC$.
  When $i \le j \le \omega$, we write $\iota_{i}^{j}  : \CC_{i} \to \CC_{j}$ for the composition of maps of that sequence.

  For each $i \le \omega$, the map $\iota_{i}^{i+1} : \CC_{i} \to \CC_{i+1}$ is a basic $I$-cellular map, specified by a set $\GenTy_{i}$ of generating types and a set $\GenTm_{i}$ of generating terms.
  For every $\bm{S} : \GenTy_{i}$, we have a boundary $\partial \bm{S} : \Ob_{\CC_{i}}$ and a generating type $\bm{S} : \yo(\iota_{i}^{i+1}(\partial \bm{S})) \to \Ty_{\CC_{i+1}}$.
  For every $\bm{f} : \GenTm_{i}$, we have a boundary $\partial \bm{f} : \Ob_{\CC_{i}}$, an output type $T \bm{f} : \yo(\partial \bm{f}) \to \Ty_{\CC_{i}}$ and a generating term $\bm{a} : (\gamma: \yo(\iota_{i}^{i+1}(\partial \bm{f}))) \to \Tm_{\CC_{i+1}}(\iota_{i}^{i+1}(T \bm{f})(\gamma))$.
  A morphism $F : \CC_{i+1} \to \CE$ is uniquely determined by the composition $F \circ \iota_{i}^{i+1}$ and by the image of the generating types and terms.

  We pose $\GenTy_{\CC} \triangleq \coprod\limits_{i<\omega} \GenTy_{i}$ and $\GenTm_{\CC} \triangleq \coprod\limits_{i<\omega} \GenTm_{i}$.
  The boundaries and output types of $\GenTy_{\CC}$  and $\GenTm_{\CC}$ are defined in the evident way using the boundaries and output types of $\GenTy_{i}$ and $\GenTm_{i}$.
\end{proof}

We can obtain an syntactic description of the general types and terms of an $I$-cellular CwF as the well-typed trees built out of the generating types and terms.
\begin{construction}\label{constr:normal_forms}
  Given an object $\Gamma : \Ob_{\CC}$, we define inductive families of sets
  \begin{alignat*}{3}
    & \NfTy && :{ } && \forall \Delta\ (\yo(\Gamma) \to \Ty_{\CC}) \to \SSet, \\
    & \Nf_{\Gamma}^{\star} && :{ } && \forall \Delta\ (\yo(\Gamma) \to \Tm_{\CC}^{\star}(\Delta)) \to \SSet, \\
    & \Nf_{\Gamma} && :{ } && \forall A\ (\yo(\Gamma) \to \Tm_{\CC}(A)) \to \SSet,
  \end{alignat*}
  generated by the following (unnamed) constructors:
  \begin{mathpar}
    \inferrule{\iS : \GenTm_{\CC} \\ \Nf^{\star}_{\Gamma}(\tau)}{\NfTy_{\Gamma}(\iS[\tau])}
    \\
    \inferrule{{}}{\Nf_{\Gamma}^{\star}(\angles{})}

    \inferrule{\Nf_{\Gamma}^{\star}(\delta) \\ \Nf_{\Gamma}(a)}{\Nf_{\Gamma}^{\star}(\angles{\delta,a})}
    \\
    \inferrule{\Var_{\Gamma}(a)}{\Nf_{\Gamma}(a)}

    \inferrule{\opf : \GenTm_{\CC} \\ \Nf^{\star}_{\Gamma}(\tau)}{\Nf_{\Gamma}(\opf[\tau])}
  \end{mathpar}
  Then for every type $A$, substitution $\sigma$ or term $a$, there is a unique element of $\NfTy(A)$, $\Nf^{\star}(\sigma)$ or $\Nf(a)$.
  In other words, types, terms and telescopes of terms admit a unique normal form.
  \lipicsEnd{}
\end{construction}
\begin{proof}[Construction]
  This is a standard normalization proof, although it is easier than usual thanks to the absence of definitional equalities.

  We first prove the existence of normal forms.
  We define a new CwF $\CC_{\nf}$; its substitutions, types and terms are those of $\CC$ equipped with normal forms.
  We omit the full definition of $\CC_{\nf}$, it is lengthy but straightforward.
  It involves the definition of the action of normal substitutions on normal forms.

  We have a projection morphism $F : \CC_{\nf} \to \CC$.
  We then construct a section $G$ of $F$, by transfinite induction on $i \le \omega$.
  The precise induction hypothesis is that for any $i \le \omega$, we construct a morphism $G_{i} : \CC_{i} \to \CC_{\nf}$ such that $F \circ G_{i} = \iota_{i}^{\omega}$.
  The zero and limit cases are straightforward, and in the successor case we only have to show that the generating types and terms admit a normal form.
  This holds essentially by definition of normal forms.
  By definition of $\CC_{\nf}$, the section $G$ equips every type $A$, term $a$ or substitution $\sigma$ with a normal form $\nf(A)$, $\nf(a)$ or $\nf(\sigma)$.

  In order to prove uniqueness, we prove that normalization is stable, \ie{} that for every normal form $A^{\nf} : \NfTy(A)$, $a^{\nf} : \Nf(a)$ or $\sigma^{\nf} : \Nf^{\star}(\sigma)$, we have $A^{\nf} = \nf(A)$, $a^{\nf} = \nf(a)$ or $\sigma^{\nf} = \nf(\sigma)$.
  This is shown by induction on normal forms.
  Most cases are straightforward.
  In the case of a generating type or term coming from the basic $I$-cellular map $\CC_{i} \to \CC_{i+1}$, we use the definition of $G_{i+1}$ on these generating types and terms.
\end{proof}

\section{Generic contexts}\label{sec:generic_contexts}

\subsection{Familially representable presheaves}\label{ssec:fam_rep}

We recall the notion of \emph{familially representable} presheaf~\cite{ConnectedLimitsFamilialRepresentabilityArtinGlueing,CorrigendaConnectedLimitsFamilialRepresentabilityArtinGlueing}.

\begin{definition}\label{def:familially_rep}
  Let $\CC$ be a category and $X : \UU$ be a presheaf over $\CC$.

  The following conditions are equivalent:
  \begin{enumerate}
  \item\label{itm:familially_rep_1} Every connected component of the category of elements $\int_{\CC} X$ is equipped with a terminal object.
    If $x : \yo(\Gamma) \to X$ is an element, the terminal object $x_{0} : \yo(\Gamma_{0}) \to X$ of its connected component is called the \defemph{most general generalization} of $x$.
  \item\label{itm:familially_rep_2} The presheaf $X$ can be decomposed as a coproduct of representable presheaves
    \[ X \simeq \coprod\limits_{i:I}(\yo(X_{i})) \]
    for some family of objects $X : I \to \Ob_{\CC}$ indexed by some set $I$.
  \item\label{itm:familially_rep_3} For every element $x : \yo(\Gamma) \to X$, we have an element $x_{0} : \yo(\Gamma_{0}) \to X$ and there is a unique morphism $f : \Gamma \to \Gamma_{0}$ such that $x = x_{0}[f]$.
    Furthermore, $x_{0}$ depends strictly naturally on $\Gamma$.
  \end{enumerate}

  When they hold, we say that $X$ is \defemph{familially representable}.
  \lipicsEnd{}
\end{definition}
\begin{proof}
  See~\cite{ConnectedLimitsFamilialRepresentabilityArtinGlueing} for the equivalence between conditions (\ref{itm:familially_rep_1}) and (\ref{itm:familially_rep_2}).
  Condition (\ref{itm:familially_rep_3}) is an unfolding of condition (\ref{itm:familially_rep_1}).
\end{proof}

\begin{definition}\label{def:locally_familially_representable}
  A dependent presheaf $Y : X \to \UU$ is \defemph{locally familially representable} if for every element $x : \yo(\Gamma) \to X$, the restricted presheaf
  \begin{alignat*}{3}
    & Y_{\mid x} && :{ } && \CPsh(\CC / \Gamma) \\
    & Y_{\mid x}(\rho : \Delta \to \Gamma) && \triangleq{ } && Y(x[\rho] : \yo(\Delta) \to X)
  \end{alignat*}
  is familially representable.
  \lipicsEnd{}
\end{definition}

Unfolding the definition, a dependent presheaf $Y : X \to \UU$ is locally familially representable if for every element $x : \yo(\Delta) \to X$, morphism $\rho : \Gamma \to \Delta$ and element $y : (\gamma : \yo(\Gamma)) \to Y(x(\rho(\gamma)))$, there is, strictly naturally in $\Gamma$, a map $\rho_{0} : \Gamma_{0} \to \Delta$ and an element $y : (\gamma : \yo(\Gamma_{0})) \to Y(x(\rho_{0}(\gamma)))$ such that there is a unique map $f : \Gamma \to \Gamma_{0}$ satisfying $\rho = \rho_{0} \circ f$ and $y = y_{0}[f]$.
The object $\Gamma_{0}$ can be seen as the extension of the context $\Delta$ that classifies the connected component of $y$.

\begin{proposition}\label{prop:lfr_telescopes}
  If a family $Y : X \to \UU$ is locally familially representable, the family of telescopes $Y^{\star} : X^{\star} \to \UU$ is also locally familially representable.
  \qed{}
\end{proposition}

\subsection{Polynomial sorts}

\begin{definition}
  Let $\CC$ be a CwF.
  We define global families $\BSort_{\CC}$ of \defemph{basic sorts}, $\MonoSort_{\CC}$ of \defemph{monomial sorts} and $\PolySort_{\CC}$ of \defemph{polynomial sorts}.
  We write $\Elem(-)$ for the elements of these families.
  Note that they do not necessarily have representable elements.
  \begin{itemize}
  \item A \defemph{basic sort} is either $\bty$ or $\btm(A)$ for some $A : \Ty_{\CC}$.
    \begin{alignat*}{3}
      & \Elem(\bty) && \triangleq{ } && \Ty_{\CC} \\
      & \Elem(\btm(A)) && \triangleq{ } && \Tm_{\CC}(A)
    \end{alignat*}
    We can view the basic sorts $\bty$ and $\btm(-)$ as \emph{codes} for the presheaves of types and terms.
  \item A \defemph{monomial sort} $[\Delta \vdash A]$ (or $[\delta:\Delta \vdash A(\delta)]$) consists of a telescope $\Delta : \Ty^{\star}_{\CC}$ and a dependent basic sort ${A : \Tm^{\star}_{\CC}(\Delta) \to \BSort_{\CC}}$.
    It represents dependent functions from $\Delta$ to $A$, or equivalently elements of $A$ in a context extended by $\Delta$.
    \begin{alignat*}{3}
      & \Elem([\Delta \vdash A]) && \triangleq{ } && (\delta : \Tm^{\star}_{\CC}(\Delta)) \to \Tm_{\CC}(A(\delta))
    \end{alignat*}
  \item A \defemph{polynomial sort} is a telescope of monomial sorts:
    \[ \PolySort_{\CC} \triangleq \MonoSort_{\CC}^{\star}.
      \lipicsEnd{}
    \]
  \end{itemize}
\end{definition}

Thus a polynomial sort is a dependent sum of dependent products of basic sorts.
Since dependent sums distribute over dependent products, $\PolySort_{\CC}$ is closed under dependent products with arities in $\Tm_{\CC}$.

The parameters of (both weakly and strictly stable) type-theoretic structures are all described by (closed)  polynomial sorts.
For instance, the parameters of an $\Id$-introduction structure are given by the closed polynomial sort
\[ \partial \Id \triangleq (A : \bty) \times (x : \btm(A)). \]
The parameters of an $\Id$-elimination structure are specified by the polynomial sort
\[ \partial \J \triangleq ((A,x) : \partial \Id) \times (P : [y : A(\gamma), p : \Id_{\Gamma}(\gamma,y) \vdash \bty]) \times (d : \btm(P(x(\gamma),\refl_{\Gamma}(\gamma)))). \]

\begin{definition}\label{def:condition_fr}
  We say that a CwF $\CC$ has \defemph{familially representable polynomial sorts} if for every closed polynomial sort $P$, the presheaf $\Elem(P)$ is familially representable.
  \lipicsEnd{}
\end{definition}

\subsection{Strictification}

\restateStrictId*
\begin{proof}
  The proof works for identity types with either a weak or a strict computation rule.

  We first consider the closed polynomial $\partial \Id \triangleq (A : \bty) \times (x : \btm(A))$.

  Let $\angles{A,x} : \yo(\Gamma) \to \Elem(\partial \Id)$ be the parameters of the stable $\Id$-introduction structure over a context $\Gamma$.
  Since $\CC$ has generic polynomial contexts, we can find a most general generalization
  $\angles{A_{0},x_{0}} : \yo(\Gamma_{0}) \to \Elem(\partial \Id)$
  of $\angles{A,x}$.
  By the universal property of $\angles{A_{0},x_{0}}$, we have a map $f : \Gamma \to \Gamma_{0}$ such that $\angles{A,x} = \angles{A_{0},x_{0}}[f]$.

  We then pose
  \begin{alignat*}{3}
    & \Id^{s}_{\Gamma}(A,x,y) && \triangleq{ } && \Id_{(\Gamma_{0},A_{0},x_{0})}[\angles{f,y}], \\
    & \refl^{s}_{\Gamma}(A,x) && \triangleq{ } && \refl_{(\Gamma_{0},A_{0},x_{0})}[f].
  \end{alignat*}
  Since most general generalizations are strictly natural, $(\Id^{s},\refl^{s})$ is a stable $\Id$-introduction structure.

  Now consider the polynomial sort
  \[ \partial \J \triangleq ((A,x) : \partial \Id) \times (P : [y : A, p : \Id^{s}(A,x,y)]\ \bty) \times (d : \btm(P(x,\refl^{s}(A,x)))). \]

  Let $\angles{A,x,P,d} : (\gamma : \yo(\Gamma)) \to \Elem(\partial \J)$ be the parameters of the stable $\Id$-elimination structure over $\Gamma$, $A$ and $x$.
  Since $\CC$ has generic polynomial contexts, we can find a most general generalization
  $\angles{A_{1},x_{1},P_{1},d_{1}} : \yo(\Gamma_{1}) \to \partial \J$.
  There is a unique map $g : \Gamma \to \Gamma_{1}$ such that $\angles{A,x,P,d} = \angles{A_{1},x_{1},P_{1},d_{1}}[g]$.

  We can also obtain the most general generalization $\angles{A_{0},x_{0}} : \yo(\Gamma_{0}) \to \partial \Id$ of $\angles{A_{1},x_{1}}$.
  We have a map $f : \Gamma_{1} \to \Gamma_{0}$ such that $\angles{A_{1},x_{1}} = \angles{A_{0},x_{0}}[f]$.
  By the universal property of most general generalizations, $\angles{A_{0},x_{0}}$ is also the most general generalization of $\angles{A,x}$.
  Thus by definition of $\Id^{s}$, we have $\Id^{s}_{\Gamma}(A,x,y) = \Id_{(\Gamma_{0},A_{0},x_{0})}[\angles{f \circ g, y}]$.

  We can finally pose
  \begin{alignat*}{3}
    & \J_{\Gamma}^{s}(A,x,P,d,y,p) && \triangleq{ } && \J_{(\Gamma_{0},A_{0},x_{0},\Gamma_{1},f,P_{1},d_{1})}[\angles{g,y,p}], \\
    & \Jbeta_{\Gamma}^{s}(A,x,P,d) && \triangleq{ } && \Jbeta_{(\Gamma_{0},A_{0},x_{0},\Gamma_{1},f,P_{1},d_{1})}[g].
  \end{alignat*}
  This determines a stable $\Id$-elimination structure $(\J^{s},\Jbeta^{s})$.
  Note that if $\Jbeta$ is strict, then $\Jbeta^{s}$ is also strict.

  By~\cref{prop:ws_wsid} the stable $\Id$-types are equivalent to the weakly stable identity types.
\end{proof}

\subsection{The local universes method}

We show that the local universes strictification method~\cite{LocalUniversesModel} factors through ours.

\begin{definition}[{\cite[Definition~3.1.3]{LocalUniversesModel}}]\label{def:condition_lf}
  A CwF $\CC$ satisfies the condition~\hyperref[def:condition_lf]{(LF)} if its underlying category has finite products, and given maps $Z \xrightarrow{g} Y \xrightarrow{f} X$, if $f$ is a display map and $g$ is either a display map or a product projection, then a dependent exponential $\Pi[f,g]$ exists.
  \lipicsEnd{}
\end{definition}
In the above definition, a display map is a finite composite of projections maps $\bm{p}_{A} : \Gamma.A \to \Gamma$; equivalently a display map is a projection map $\bm{p}_{\Delta} : \Gamma.\Delta \to \Gamma$ where $\Gamma$ is an object of $\CC$ and $\Delta$ is a telescope over $\Gamma$.

Condition~\hyperref[def:condition_lf]{(LF)} can essentially be unfolded into the following two representability conditions:
\begin{itemize}
\item For every object $\Gamma : \Ob_{\CC}$, telescope $\Delta : \yo(\Gamma) \to \Ty_{\CC}^{\star}$ and object $\Theta : \Ob_{\CC}$, the presheaf
  \[ (\gamma: \yo(\Gamma)) \times (\Tm_{\CC}^{\star}(\Delta(\gamma)) \to \yo(\Theta)) \]
  is representable.

\item For every object $\Gamma : \Ob_{\CC}$, telescope $\Delta : \yo(\Gamma) \to \Ty_{\CC}^{\star}$ and type
  \[ A : (\gamma : \yo(\Gamma)) (\delta : \Tm_{\CC}^{\star}(\Delta(\gamma))) \to \Ty_{\CC}, \]
  the presheaf
  \[ (\gamma: \yo(\Gamma)) \times ((\delta : \Tm_{\CC}^{\star}(\Delta(\gamma))) \to \Tm_{\CC}(A(\gamma,\delta))) \]
  is representable.
\end{itemize}

\begin{definition}
  Let $\CC$ be a CwF.

  A \defemph{local universe} is a pair $(V,E)$, where $V : \Ob_{\CC}$ is an object of $\CC$ and $E : \yo(V) \to \Ty_{\CC}$ is a type over $V$.

  The \defemph{local universe model} $\CC_{!}$ is another CwF over the same base category.
  We write $(\Ty_{!},\Tm_{!})$ for its family of types and terms.

  A type of $\CC_{!}$ is a triple $(V,E,\chi)$, where $(V,E)$ is a local universe, and $\chi : \yo(V)$.
  There is a natural transformation $\Ty_{!} \to \Ty_{\CC}$, sending $(V,E,\chi)$ to $E(\chi)$.

  The terms of $\CC_{!}$ are induced by this natural transformation: $\Tm_{!}(V,E,\chi) \triangleq \Tm_{\CC}(E(\chi))$.
  The local representability of the dependent presheaf $\Tm_{!}$ follows from the local representability of $\Tm_{\CC}$.
  \lipicsEnd{}
\end{definition}

There is a CwF morphism $\CC_{!} \to \CC$ lying over the identity functor.
That morphism is surjective on types and bijective on terms.
In particular, it is a trivial fibration.

Any weakly stable type-theoretic structure can be lifted along $\CC_{!} \to \CC$.
Since $\CC_{!} \to \CC$ is injective on terms, definitional equalities between terms can also be lifted.
It is however not generally possible to lift definitional equalities between types.

\begin{proposition}\label{prop:lf_famrep}
  If $\CC$ satisfies condition~\hyperref[def:condition_lf]{(LF)}, then $C_{!}$ has familially representable polynomial sorts.
\end{proposition}
\begin{proof}
  We prove by induction on closed polynomial sorts that for every $P : \PolySort_{\CC_{!}}$, the presheaf $\Elem(P)$ is familially representable.
  \begin{description}
  \item[Case $P = \diamond$:] \hfill \\
    Then $\Elem(P)$ is the terminal presheaf, which is represented by the terminal object of $\CC$.
  \item[Case $P = Q. M$:] \hfill \\
    Here $M : \Elem(Q) \to \MonoSort_{\CC_{!}}$ is a monomial sort over $Q$.

    Take an element $\angles{q,a} : \yo(\Gamma) \to \Elem(Q. M)$.
    Our goal is to construct the most general generalization of $\angles{q,a}$, \ie{} a terminal object of the connected component of $\angles{q,a}$ in the category of elements of $\Elem(Q.M)$.

    By the induction hypothesis, we have a most general generalization $q_{0} : \yo(\Gamma_{0}) \to \Elem(Q)$ of $q$.
    By its universal property, there is a unique map $f : \Gamma \to \Gamma_{0}$ such that $q = q_{0}[f]$.

    We now inspect $M[q_{0}] : \yo(\Gamma_{0}) \to \MonoSort_{\CC_{!}}$, noting that $a : (\gamma : \yo(\Gamma)) \to \Elem(M[q_{0}](f(\gamma)))$.
    \begin{description}
    \item[Case ${M[q_{0}] = \lambda \gamma \mapsto [\Delta(\gamma) \vdash \bty]}$:] \hfill \\
      Here $\Delta : \yo(\Gamma_{0}) \to \Ty_{!}^{\star}$ is a telescope over $\Gamma_{0}$.

      We know that $a : (\gamma:\yo(\Gamma)) \to \Tm^{\star}_{!}(\Delta(f(\gamma))) \to \Ty_{!}$.
      By definition of the presheaf $\Ty_{!}$, this means that we have a local universe $(V,E)$ and a classifying map
      \[ \chi : (\gamma:\yo(\Gamma)) \to \Tm^{\star}_{!}(\Delta(f(\gamma))) \to \yo(V) \]
      such that
      $a = \lambda (\gamma,\delta) \mapsto E(\chi(\gamma,\delta))$.

      By condition~\hyperref[def:condition_lf]{(LF)}, there exists an object $\Gamma_{1}$ representing the presheaf
      \[ (\gamma : \yo(\Gamma_{0})) \times (v : \Tm^{\star}_{!}(\Delta(f(\gamma))) \to \yo(V)). \]

      We now define $\angles{q_{1},a_{1}} : \yo(\Gamma_{1}) \to \Elem(Q. M)$:
      \begin{alignat*}{3}
        & q_{1}(\gamma,v) && \triangleq{ } && q_{0}(\gamma), \\
        & a_{1}(\gamma,v) && \triangleq{ } && \lambda (\delta : \Tm^{\star}_{!}(\Delta(f(\gamma)))) \mapsto E(v(\delta)).
      \end{alignat*}

      We have $\angles{q,a} = \angles{q_{1},a_{1}}[\angles{f, \chi}]$.
      By the universal properties of $\Gamma_{1}$ and $q_{0}$, the element $\angles{q_{1},a_{1}}$ is the most general generalization of $\angles{q,a}$.

    \item[Case ${M[q_{0}] = \lambda \gamma \mapsto [\delta:\Delta(\gamma) \vdash \btm(A(\gamma,\delta))]}$:] \hfill \\
      Here $\Delta : \yo(\Gamma_{0}) \to \Ty_{!}^{\star}$ is a telescope over $\Gamma_{0}$ and $A : (\gamma : \yo(\Gamma_{0})) \to \Tm^{\star}_{!}(\Delta(\gamma)) \to \Ty_{!}$.
      We can decompose $A$ into a local universe $(V,E)$ and a classifying map
      \[ \chi : (\gamma : \yo(\Gamma_{0})) (\delta : \Tm^{\star}_{!}(\Delta(\gamma))) \to \yo(V) \]
      such that $A = \lambda(\gamma,\delta) \mapsto E(\chi(\gamma,\delta))$.

      We know that $a : (\gamma : \yo(\Gamma)) (\delta : \Tm^{\star}_{!}(\Delta[f])) \to \Tm_{\CC}(E(\chi(f(\gamma),\delta)))$.

      By condition~\hyperref[def:condition_lf]{(LF)}, there exists an object $\Gamma_{1}$ representing the presheaf
      \[ (\gamma : \yo(\Gamma_{0})) \times (x : (\delta : \Tm^{\star}_{!}(\Delta(f(\gamma)))) \to \Tm_{\CC}(E(\chi(\gamma,\delta)))). \]

      We now define $\angles{q_{1},a_{1}} : \yo(\Gamma_{1}) \to \Elem(Q. M)$:
      \begin{alignat*}{3}
        & q_{1}(\gamma,x) && \triangleq{ } && q_{0}(\gamma), \\
        & a_{1}(\gamma,x) && \triangleq{ } && \lambda (\delta : \Tm^{\star}_{!}(\Delta(f(\gamma)))) \mapsto x(\delta).
      \end{alignat*}

      We have $\angles{q,a} = \angles{q_{1},a_{1}}[\angles{f, a}]$.
      By the universal properties of $\Gamma_{1}$ and $q_{0}$, the element $\angles{q_{1},a_{1}}$ is the most general generalization of $\angles{q,a}$.
      \qedhere{}
    \end{description}
  \end{description}
\end{proof}

\section{Most general generalizations in free CwFs}\label{sec:mgg_free_cwfs}

In this section we prove the following result.
\restateMggFree*

We fix an $I$-cellular CwF $\CC$.
We use the explicit description of the types and terms of $\CC$ that was given in~\cref{constr:free_gens}.

\subsection{First-order unification}

First-order unification~\cite{RobinsonUnification,GoguenWhatIsUnification} is usually presented for free unityped or simply typed theories, but it is folklore that the same unification procedure is also valid for free dependently typed theories\footnote{This is observed by Simon Henry in \url{https://mathoverflow.net/questions/307373/on-a-surprising-property-of-free-theories}.}, \ie{} for freely generated CwFs.
In our setting, this means that the category of cones over any pair of parallel substitutions is either empty or has a terminal object, which is then the \emph{most general unifier} of the two substitutions.

We prove a slightly stronger result, for contexts that are split into \emph{flexible} and \emph{rigid} parts.
The unification procedure can only change the flexible part.

\begin{definition}
  An \defemph{unification context} is an object of the form $\Gamma.\Delta$, where $\Delta$ is a telescope over $\Gamma$.
  The variables of $\Gamma$ are called \defemph{flexible variables}, while the variables from $\Delta$ are called \defemph{rigid variables}.

  A morphism of unification contexts is a substitution that preserves the rigid variables, \ie{} a substitution of the form $\rho^{+} : \Theta.\Delta[\rho] \to \Gamma.\Delta$ for some $\rho : \Theta \to \Gamma$.
  \lipicsEnd{}
\end{definition}

\begin{definition}[Unifiers]
  Let $\Gamma.\Delta$ be a unification context and $X$ be a dependent presheaf over $\yo(\Gamma.\Delta)$.
  A \defemph{unifier} of a pair $a,b : (x : \yo(\Gamma.\Delta)) \to X(x)$ of parallel elements of $X$ is a morphism $\rho : \Theta \to \Gamma$ such that $a[\rho^{+}] = b[\rho^{+}]$.
  We say that $a$ and $b$ are \defemph{unifiable} if there merely exists a unifier.

  A \defemph{most general unifier} is a terminal unifier.
  \lipicsEnd{}
\end{definition}

\begin{lemma}[Instantiation]\label{constr:inst}
  Let $\Gamma$ be a context, $a : (\gamma:\yo(\Gamma)) \to \Tm_{\CC}(A(\gamma))$ be a variable from $\Gamma$ and $b : (\gamma:\yo(\Gamma)) \to \Tm_{\CC}(A(\gamma))$ be a term of type $A$ such that $b \neq a$.

  If the terms $a$ and $b$ are unifiable, then we can construct a most general unifier $\Gamma[a := b]$.
  Moreover, the length of $\Gamma[a := b]$ is less than the length of $\Gamma$.
\end{lemma}
\begin{proof}
  We have a bijective renaming $\Gamma \simeq (\gamma_{0}:\Gamma_{0}). \Gamma_{1}(\gamma_{0})$ where $\Gamma_{0}$ is the support of the term $b$.
  Up to this renaming, we have $A : \yo(\Gamma_{0}) \to \Ty_{\CC}$, $b : (\gamma_{0}:\yo(\Gamma_{0})) \to \Tm_{\CC}(A(\gamma_{0}))$
  and $a : (\gamma_{0}:\yo(\Gamma_{0}))(\gamma_{1}:\Tm_{\CC}^{\star}(\Gamma_{1}(\gamma_{0}))) \to \Tm_{\CC}(A(\gamma_{0}))$.

  The variable $a$ cannot belong to the support $\Gamma_{0}$ of $b$, since $a$ and $b$ are unifiable and different; this is the \emph{occurs check} of first-order unification.
  Indeed, assuming that $a$ did belong to $\Gamma_{0}$ and considering the unifier $\rho$ of $a$ and $b$, the term $b[\rho]$ would be infinite.

  Thus $a$ is a variable from $\Gamma_{1}$ and we can write $\Gamma_{1}(\gamma_{0}) = (\gamma_{2}:\Gamma_{2}(\gamma_{0})).(a:A(\gamma_{0})).\Gamma_{3}(\gamma_{0},\gamma_{2},a)$.

  We now pose $\Gamma[a := b] \triangleq (\gamma_{0}:\Gamma_{0}).(\gamma_{2}:\Gamma_{2}(\gamma_{0})).\Gamma_{3}(\gamma_{0},\gamma_{2},b)$.
  It is the most general unifier of $a$ and $b$.
\end{proof}

\begin{lemma}[Strengthening]\label{constr:strengthening}
  Let $\Gamma.\Delta$ be a unification context, $a : (\gamma:\yo(\Gamma)) \to \Tm_{\CC}(A(\gamma))$ be a term over $\Gamma$ and $b : ((\gamma,\delta):\yo(\Gamma.\Delta)) \to \Tm_{\CC}(A(\gamma))$ be a term of type $A[\bm{p}_{\Delta}]$.

  If the terms $a[\bm{p}_{\Delta}]$ and $b$ are unifiable, then there exists a (necessarily unique) term $b' : (\gamma:\yo(\Gamma)) \to \Tm_{\CC}(A(\gamma))$ such that $b = b'[\bm{p}_{\Delta}]$.
\end{lemma}
\begin{proof}
  Let $\rho : \Omega \to \Gamma$ be a unifier of $a[\bm{p}_{\Delta}]$ and $b$.
  Then $b[\rho^{+}] = a[\bm{p}_{\Delta}][\rho^{+}] = a[\rho][\bm{p_{\Delta[\rho]}}]$.
  Thus $b[\rho^{+}]$ cannot depend on any variable from $\Delta[\rho]$.
  Since $\rho^{+}$ preserves the variables of $\Delta$, the term $b$ cannot depend on any variable from $\Delta$.
  Therefore it can be strengthened to some term $b' : (\gamma:\yo(\Gamma)) \to \Tm_{\CC}(A(\gamma))$.
\end{proof}

\begin{theorem}[First-order unification]\label{thm:mgu_flex_rigid}
  Let $\Gamma.\Delta$ be a unification context and $X : \yo(\Gamma.\Delta) \to \UU$ a dependent presheaf of the form $\Tm^{\star}_{\CC}(\Xi)$, $\Ty_{\CC}$ or $\Tm_{\CC}(A(-))$.
  If there exists a unifier $\sigma : \Theta \to \Gamma$ of a pair $x_{1},x_{2} : ((\gamma,\delta) : \yo(\Gamma.\Delta)) \to X(\gamma,\delta)$ of parallel elements of $X$, then there exists a most general unifier $\rho : \Omega \to \Gamma$, such that either $\rho = \id$ or the length of $\Omega$ is less than the length of $\Gamma$.
  \qed{}
\end{theorem}
\begin{proof}
  By nested inductions first on the length of $\Gamma$, and then on the normal form of the substitution, type, or term $x_{1}$.
  \begin{description}
  \item[Case $({\Gamma = \diamond})$:]
    Let $\sigma : \Theta \to \diamond$ be a unifier of $x_{1}$ and $x_{2}$.
    The map $\sigma = \angles{}$ is an epimorphism.

    Thus $x_{1} = x_{2}$ and $\id : \Gamma \to \Gamma$ is the most general unifier of $x_{1}$ and $x_{2}$.

  \item[Case $({X = \Tm^{\star}_{\CC}(\diamond)})$:] \hfill \\
    In that case, $x_{1} = x_{2} = \angles{}$ and $\id : \Gamma \to \Gamma$ is the most general unifier of $x_{1}$ and $x_{2}$.

  \item[Case $({X = \Tm^{\star}_{\CC}(\Xi.A)})$:] \hfill \\
    We can write $x_{1} = \angles{\xi_{1},a_{1}}$ and $x_{2} = \angles{\xi_{2},a_{2}}$.
    By the induction hypothesis for $\xi_{1}$, we have a most general unifier $\rho : \Omega \to \Gamma$ of $\xi_{1}$ and $\xi_{2}$.

    If $\rho = \id$, then $a_{1}$ and $a_{2}$ are parallel terms and by the induction hypothesis for $a_{1}$ we can find a most general unifier $\rho' : \Omega' \to \Gamma$ of $a_{1}$ and $a_{2}$.
    It is then also a most general unifier of $x_{1}$ and $x_{2}$.

    Otherwise, the length of $\Omega$ is less than the length of $\Gamma$.
    By the induction hypothesis for $\Omega$, we can then find a most general unifier $\rho' : \Omega' \to \Omega$ of $a_{1}[\rho]$ and $a_{2}[\rho]$.
    The composite $(\rho \circ \rho') : \Omega' \to \Gamma$ is then a most general unifier of $x_{1}$ and $x_{2}$.

  \item[Case $({X = \Ty_{\CC}})$:] \hfill \\
    We can write $x_{1} = \iS[\sigma_{1}]$ for some generating type $\iS$ and $\sigma_{1} : \Gamma.\Delta \to \partial \iS$.
    Since $x_{1}$ and $x_{2}$ are unifiable, we can also write $x_{2} = \iS[\sigma_{2}]$ for some $\sigma_{2} : \Gamma.\Delta \to \partial \iS$.
    By the induction hypothesis for $\sigma_{1}$, we have a most general unifier of $\sigma_{1}$ and $\sigma_{2}$.
    It is then also a most general unifier of $x_{1}$ and $x_{2}$.

  \item[Case $({X = \Tm_{\CC}(A(-))})$:] \hfill \\
    We have several subcases depending on the parallel terms $x_{1}$ and $x_{2}$.
    \begin{description}
    \item[Case ${(x_{1} = \opf[\sigma_{1}])}$ and ${(x_{2} = \opg[\sigma_{2}])}$:] \hfill \\
      Since $x_{1}$ and $x_{2}$ are unifiable, $\opf = \opg$.
      Here $\sigma_{1} : \Gamma.\Delta \to \partial \opf$ and $\sigma_{2} : \Gamma.\Delta \to \partial \opf$.
      By the induction hypothesis for $\sigma_{1}$, we have a most general unifier of $\sigma_{1}$ and $\sigma_{2}$.
      It is then also a most general unifier of $x_{1}$ and $x_{2}$.
    \item[If either $x_{1}$ or $x_{2}$ is a variable from $\Gamma$:] \hfill \\
      Without loss of generality, assume that $x_{1}$ is a variable from $\Gamma$.
      By~\cref{constr:strengthening}, the term $x_{2}$ can be strengthened to only depend on $\Gamma$.
      If $x_{1} = x_{2}$ then $\id : \Gamma \to \Gamma$ is the most general unifier of $x_{1}$ and $x_{2}$.
      Otherwise $x_{1} \neq x_{2}$ and the instantiation $\Gamma[x_{1} := x_{2}]$ is the most general unifier of $x_{1}$ and $x_{2}$, by~\cref{constr:inst}.
      The length of $\Gamma[x_{1} := x_{2}]$ is then less than the length of $\Gamma$.
    \item[Otherwise, both $x_{1}$ and $x_{2}$ are variables from $\Delta$:] \hfill \\
      Since $x_{1}$ and $x_{2}$ are unifiable by a substitution that preserves the variables from $\Delta$, they have to be equal.
      Then $\id : \Gamma \to \Gamma$ is the most general unifier of $x_{1}$ and $x_{2}$.
      \qedhere{}
    \end{description}
  \end{description}
\end{proof}

\begin{remark}
  Note that~\cref{thm:mgu_flex_rigid} implies that the families
  \begin{alignat*}{3}
    & (\Delta : \Ty^{\star}_{\CC}) \times (f,g : \Tm^{\star}_{\CC}(\Delta) \to \yo(\Xi)) && \mapsto{ } && f = g \\
    & (\Delta : \Ty^{\star}_{\CC}) \times (A,B : \Tm^{\star}_{\CC}(\Delta) \to \Ty_{\CC}) && \mapsto{ } && A = B \\
    & (\Delta : \Ty^{\star}_{\CC}) \times (A : \Tm^{\star}_{\CC}(\Delta) \to \Ty_{\CC}) \times (a,b : \forall \delta \to \Tm_{\CC}(A(\delta))) && \mapsto{ } && a = b
  \end{alignat*}
  are locally familially representable.
  Indeed, their categories of elements are the categories of unifiers for substitutions, types or terms.
  By~\cref{thm:mgu_flex_rigid}, these categories are either empty, or admit a terminal object.
  In particular, every connected component admits a terminal object.
  \lipicsEnd{}
\end{remark}

\subsection{Most general generalizations}

We now apply first-order unification to the construction of most general generalizations.

We first describe this construction informally.
For any type $B$ over a unification context $\Gamma. \Delta$, we compute some $B_{0}$ over a context of the form $\Gamma_{0}. \Delta_{0}$ and a substitution $f : \Gamma \to \Gamma_{0}$ such that $\Delta = \Delta_{0}[f]$ and $B = B_{0}[f^{+}]$.
The type $B_{0}$ should be the \emph{most general generalization} of $B$ that retains the dependency on $\Delta$.

The type $B_{0}$ is essentially obtained by removing the dependencies on $\Gamma$, that is by replacing the subterms of $B$ that only depend on $\Gamma$ by new variables; these new variables are collected in the new context $\Gamma_{0}$.
Because of the dependencies of the generating terms, it is not always possible to fully remove a subterm.
We have to rely on first-order unification to determine which parts can be removed; some of the new variables may need to be instantiated to more precise terms.

We give examples involving the following generating types and terms.
\begin{alignat*}{3}
  & \iX && :{ } && \Ty \\
  & \iY && :{ } && \Tm(\iX) \to \Ty \\
  & \inner{f_{1}} && :{ } && \Tm(\iX) \to \Tm(\iX) \\
  & \inner{f_{2}} && :{ } && \Tm(\iX) \to \Tm(\iX) \to \Tm(\iX) \\
  & \inner{g} && :{ } && \forall (x : \Tm(\iX))\ (y : \Tm(\iY(x))) \to \Tm(\iX) \\
  & \inner{h} && :{ } && \forall (x : \Tm(\iX)) \to \Tm(\iY(x)) \to \Tm(\iY(\inner{f_{1}}(x)))
\end{alignat*}

We write $x,y,z,\dots$ for the variables from $\Gamma$ and $\overline{x},\overline{y},\overline{z},\dots$ for the variables from $\Delta$.

\begin{itemize}
\item Consider $B = \iY(\inner{f_{1}}(x))$ over $(x : \iX)$. \\
  Then we can pose $B_{0} = \iY(y)$ over $(y : \iX)$, we have $B = B_{0}[y \mapsto \inner{f_{1}}(x)]$.
\item Consider $B = \iY(\inner{f_{1}}(\overline{x}))$ over $(\overline{x} : \iX)$. \\
  Then we have to keep $B_{0} = \iY(\inner{f_{1}}(\overline{x}))$.
\item Consider $B = \iY(\inner{f_{2}}(\inner{f_{1}}(x), \inner{f_{1}}(\overline{y})))$ over $( x : \iX, \overline{y} : \iX)$. \\
  Then $B_{0} = \iY(\inner{f_{2}}(z, \inner{f_{1}}(\overline{y})))$ over $(z : \iX, \overline{y} : \iX)$; we have $B = B_{0}[z \mapsto \inner{f_{1}}(x)]$.
\item Consider $B = \iY(\inner{g}(\inner{f_{1}}(x), \overline{y}))$ over $(x : \iX, \overline{y} : \iY(\inner{f_{1}}(x)))$. \\
  Then $B_{0} = \iY(\inner{g}(z,\overline{y}))$ over $(z : \iX, \overline{y} : \iY(z))$; we have $B = B_{0}[z \mapsto \inner{f_{1}}(x), \overline{y} \mapsto \overline{y}]$.
\item Consider $B = \iY(\inner{g}(\inner{f_{1}}(x), \inner{h}(x,\overline{y})))$ over $(x : \iX, \overline{y} : \iY(x))$. \\
  The $B_{0} = B$.
  We cannot prune the subterm $\inner{f_{1}}(x)$, because of the typing constraints of $\inner{g}$ and $\inner{h}$.
\end{itemize}

\begin{proposition}\label{prop:free_mono_lfr}
  The families
  \begin{alignat}{3}
    \label{itm:free_mono_lfr_subst}
    & (\Delta : \Ty^{\star}_{\CC}) \times (\Xi : \Ob_{\CC}) && \mapsto{ } &&
                                                                             (\Tm_{\CC}^{\star}(\Delta) \to \Tm_{\CC}^{\star}(\Xi)) \\
    \label{itm:free_mono_lfr_ty}
    & (\Delta : \Ty^{\star}_{\CC}) && \mapsto{ } &&
                                                    (\Tm_{\CC}^{\star}(\Delta) \to \Ty_{\CC}) \\
    \label{itm:free_mono_lfr_tm}
    & (\Delta : \Ty^{\star}_{\CC}) \times (A : \Tm_{\CC}^{\star}(\Delta) \to \Ty_{\CC}) && \mapsto{ } &&
                                                                                                         \forall (\delta : \Tm_{\CC}^{\star}(\Delta)) \to \Tm_{\CC}(A(\delta))
  \end{alignat}
  are locally familially representable.

  In particular the family $\MonoSort_{\CC}$, which is the coproduct of the families (\ref{itm:free_mono_lfr_ty}) and (\ref{itm:free_mono_lfr_tm}), is locally familially representable.
\end{proposition}
\begin{proof}
  The local familial representability can be unfolded to the following conditions:

  Fix the following data:
  \begin{itemize}
  \item An object $\Gamma : \Ob_{\CC}$;
  \item An object $\Omega$, a telescope $\Delta : \yo(\Omega) \to \Ty^{\star}_{\CC}$ and a map $\rho : \Gamma \to \Omega$;
  \item Either:
    \begin{itemize}
    \item An object $\Xi$ and a map $\xi : \Gamma.\Delta[\rho] \to \Xi$;
    \item A type $A : \yo(\Gamma.\Delta[\rho]) \to \Ty_{\CC}$;
    \item A type $A : \yo(\Omega.\Delta) \to \Ty_{\CC}$ and a term $a : (x : \yo(\Gamma.\Delta[\rho])) \to \Tm_{\CC}(A(\rho^{+}(x)))$
    \end{itemize}
  \end{itemize}
  Then we have to construct the following components, strictly naturally in $\Gamma$:
  \begin{itemize}
  \item An object $\Gamma_{0} : \Ob_{\CC}$;
  \item A map $\rho_{0} : \Gamma_{0} \to \Omega$;
  \item Either:
    \begin{itemize}
    \item A map $\xi_{0} : \Gamma_{0}.\Delta[\rho_{0}] \to \Xi$;
    \item A type $A_{0} : \yo(\Gamma_{0}.\Delta[\rho_{0}]) \to \Ty_{\CC}$;
    \item A term $a_{0} : (x : \yo(\Gamma_{0}.\Delta[\rho_{0}])) \to \Tm_{\CC}(A(\rho_{0}^{+}(x)))$;
    \end{itemize}
  \item Such that there exists a unique map $f : \Gamma \to \Gamma_{0}$ satisfying $\rho = \rho_{0}[f]$ and $\xi = \xi_{0}[f^{+}]$, $A = A_{0}[f^{+}]$ or $a = a_{0}[f^{+}]$.
  \end{itemize}

  \begin{mathpar}
    \begin{tikzcd}[row sep=40pt]
      &
      &
      \Xi
      \\
      \Gamma.\Delta[\rho]
      \ar[rr, "\rho^{+}", bend right]
      \ar[r, dashed, "f^{+}"]
      \ar[d, "\bm{p}_{\Delta[\rho]}"]
      \ar[rru, "\xi", bend left]
      &
      \Gamma_{0}.\Delta[\rho_{0}]
      \ar[r, dashed, "\rho_{0}^{+}"]
      \ar[ru, dashed, "\xi_{0}"]
      &
      \Omega.\Delta
      \ar[d, "\bm{p}_{\Delta}"]
      \\
      \Gamma
      \ar[rr, "\rho", bend right]
      \ar[r, dashed, "f"]
      &
      \Gamma_{0}
      \ar[r, dashed, "\rho_{0}"]
      &
      \Omega
    \end{tikzcd}
  \end{mathpar}

  We construct the most general generalizations by induction on the normal forms of $\xi$, $A$ or $a$.
  The strict naturality in $\Gamma$ will be proven in a second step.

  \begin{description}
  \item[Case $(\Xi = \diamond)$ and $(\xi = \angles{})$:] \hfill \\
    We pose $\Gamma_{0} = \Omega$, $\rho_{0} = \id$, $\xi_{0} = \angles{}$ and $f = \rho$.

  \item[Case $(\Xi = \Theta.A)$ and $(\xi = \angles{\theta,a})$:] \hfill \\
    In that case $\theta : \Gamma.\Delta[\rho] \to \Theta$ and $a : (x : \yo(\Gamma.\Delta[\rho])) \to A(\theta(x))$.

    By the induction hypothesis for $\theta$, we have $\Gamma_{0}$, $\rho_{0} : \Gamma_{0} \to \Omega$, $\theta_{0} : \Gamma_{0}.\Delta[\rho_{0}] \to \Theta$ and there exists a unique map $f : \Gamma \to \Gamma_{0}$ such that $\rho_{0}[f] = \rho$ and $\theta_{0}[f^{+}] = \theta$.

    By the induction hypothesis for $a$, we have $\Gamma_{1}$, $\rho_{1} : \Gamma_{1} \to \Gamma_{0}$, $a_{1} : \Gamma_{1}.\Delta[\rho_{0}][\rho_{1}]$ and there is a unique map $g : \Gamma \to \Gamma_{1}$ such that $\rho_{1}[g] = f$ and $a_{1}[g^{+}] = a$.

    We then pose $\Gamma_{2} = \Gamma_{1}$, $\rho_{2} = \rho_{0} \circ \rho_{1}$ and $\xi_{2} = \angles{\theta_{0}[\rho_{1}], a_{1}}$.
    The map $g : \Gamma \to \Gamma_{1}$ is then the unique map such that $\rho_{2}[g] = \rho$ and $\xi_{2}[g^{+}] = \xi$.

  \item[Case ${(A = \iS[\sigma])}$:] \hfill \\
    Here $\sigma : \Gamma \to \partial \iS$.
    We just use the induction hypothesis for $\sigma$, and pose $A_{0} = \iS[\sigma_{0}]$.

  \item[Case ${(a = a'[\bm{p}_{\Delta[\rho]}])}$] \hfill \\
    As a special case, we check if the term $a$ depends on any variable from $\Delta[\rho]$.
    If it can be strengthened to a term $a'$ over $\Gamma$ such that $a'[\bm{p}_{\Delta[\rho]}] = a$, we also know that the type $A$ cannot depend on any variable from $\Delta$, and can be strengthened to $A' : \yo(\Omega) \to \Ty_{\CC}$ such that $A'[\bm{p}_{\Delta}] = A$.
    We then pose $\Gamma_{0} = (\omega:\Omega).(a_{0}:A'(\delta))$, $\rho_{0} = \bm{p}_{A'} : \Gamma_{0} \to \Omega$ and $f = \angles{\rho,a'}$.

  \item[Case ${(\Var_{\Gamma.\Delta[\rho]}(a))}$:] \hfill \\
    If $a$ is a variable from $\Gamma.\Delta[\rho]$, then $a$ has to be variable from $\Delta[\rho]$, as variables from $\Gamma$ are dealt with in the case above.

    Then we let $a_{0}$ be the corresponding variable from $\Delta$ and we pose $\Gamma_{0} = \Omega$, $\rho_{0} = \id$ and $f = \rho$.

  \item[Case ${(a = \opf[\tau])}$:] \hfill \\
    Here $\tau : \Gamma.\Delta[\rho] \to \partial \opf$ and $a : (x : \yo(\Gamma.\Delta[\rho])) \to \Tm_{\CC}(T \opf(\tau(x)))$.
    We then know that $A[\rho^{+}] = T \opf[\tau]$.

    By the induction hypothesis for $\tau$, we have $\Gamma_{0}$, $\rho_{0} : \Gamma_{0} \to \Omega$, $\tau_{0} : \Gamma_{0}.\Delta[\rho_{0}] \to \partial \opf$ and there is a unique map $f : \Gamma \to \Gamma_{0}$ such that $\rho = \rho_{0}[f]$ and $\tau = \tau_{0}[f^{+}]$.

    The types $A[\rho_{0}^{+}]$ and $T\opf[\tau_{0}]$ may differ.
    We know however that they are unifiable by the map $f^{+}$; thus by first-order unification~(\cref{thm:mgu_flex_rigid}), we can find a most general unifier $\rho_{1} : \Gamma_{1} \to \Gamma_{0}$ of these two types.
    By the universal property of the most general unifier, we have a factorization of $f$ as a map $g : \Gamma \to \Gamma_{1}$ followed by $\rho_{1} : \Gamma_{1} \to \Gamma_{0}$.

    Now we pose $\Gamma_{2} = \Gamma_{1}$, $\rho_{2} = \rho_{0} \circ \rho_{1}$, $a_{2} = \opf[\tau[\rho_{1}^{+}]]$.
    The map $g : \Gamma_{1} \to \Gamma_{0}$ is then the unique map such that $\rho_{2}[g] = \rho$ and $a_{2}[g] = \tau$.
  \end{description}

  It remains to prove that the above construction is strictly natural in $\Gamma$: we have to prove for any $\xi$, $A$ or $a$ and any substitution $\sigma : \Lambda \to \Gamma$ that the most general generalizations of $\xi$ and $\xi \circ \sigma$ (or $A$ and $A[\sigma^{+}]$, or $a$ and $a[\sigma^{+}]$) coincide.
  We prove this by induction on the normal forms of $\xi$, $A$ or $a$, following the inductive cases of the previous construction.
  It is then straightforward to check that the construction follows the same cases for both $\xi$ and $\xi[\sigma^{+}]$ (or $A$ and $A[\sigma^{+}]$, or $a$ and $a[\sigma^{+}]$).

  The main subtlety happens when $a : \yo(\Gamma.\Delta[\rho]) \to \Tm_{\CC}(-)$ is a variable from $\Gamma$.
  In that case, the substituted term $a[\sigma^{+}] : \yo(\Lambda.\Delta[\rho][\sigma]) \to \Tm_{\CC}(-)$ is not necessarily a variable.
  However it can be strengthened to a term that only depends on $\Lambda$.
  Thus our construction of the most general generalization of both $a$ and $a[\sigma^{+}]$ will use the special case for terms that don't depend on $\Delta$.
  Without this special case, we would not be able to prove that our construction is strictly natural in $\Gamma$.
  \qedhere{}
\end{proof}

\begin{proof}[Proof of~\cref{thm:mgg_free}]
  This follows from~\cref{prop:free_mono_lfr} and~\cref{prop:lfr_telescopes}.
\end{proof}

\subsection{Strictification}

\restateLeftStrictification*
\begin{proof}
  Let $\CD$ be an $I$-cellular replacement of $\CC$.
  We have a trivial fibration $F : \CD \to \CC$ in $\CCwf$.
  By~\cref{prop:tfib_lift_id}, $\CD$ can be equipped with weakly stable identity types $\Id$ that are strictly preserved by $F$.

  By~\cref{thm:mgg_free}, $\CD$ has familially representable polynomials sorts.
  Thus by~\cref{thm:lfr_poly_strict_id}, $\CD$ has stable identity types $\Id^{s}$ that are weakly equivalent to the weakly stable identity types.
  In other words, the CwF morphism $\id : (\CD,\Id^{s}) \to (\CD,\Id)$ weakly preserves identity types.
  Then the composition $(\CD,\Id^{s}) \xrightarrow{\id} (\CD,\Id) \xrightarrow{F} (\CC,\Id)$ weakly preserves identity types.
\end{proof}

\section{Other type-theoretic structures}\label{sec:other_structures}

So far we have only considered (weak) identity types.
However our methods can more generally be applied to any weakly stable weak type-theoretic structure.
Indeed the proofs of~\cref{thm:lfr_poly_strict_id} and \cref{thm:left_strictification_id} only rely on~\cref{prop:ws_wsid} and on the fact that the parameters of the identity introduction and elimination structures can be specified by (closed) polynomial sorts.
Thus the same proof scheme works for any type-theoretic structure that is weakly stable (in the sense that it satisfies a variant of~\cref{prop:ws_wsid}).
This holds in particular for most standard type-theoretic structures, including $\Pi$-types, $\Sigma$-types, coproducts, natural numbers and other inductive types, \etc.

Note that in general, weak structures can only be specified in presence of identity types; thus their strictification depends on the strictification of identity types.
It is then necessary to see~\cref{thm:lfr_poly_strict_id} as a construction.

\section{Towards full coherence theorems}\label{sec:towards_coherence}

We have presented general strictification methods for weakly stable weak type-theoretic structures.
However we generally want coherence theorems that give a more precise comparison between the categories $\CCwf^{\cxl}_{s}$ and $\CCwf^{\cxl}_{ws}$ of contextual CwFs equipped with stable or weakly stable weak type-theoretic structures (for some unspecified choice of such structures).

Following~\cite{HomotopyTheoryTTs,IsaevModelStructuresOnModels}, we expect that these categories can be equipped with cofibrantly generated left-semi model structures, with trivial fibrations as defined in~\cref{def:trivial_fibration}.
We then want to prove that the free-forgetful adjunction
\begin{mathpar}
  \begin{tikzcd}
    \CCwf^{\cxl}_{ws}
    \ar[r,"L",bend left]
    \ar[r,phantom,"\bot"]
    &
    \CCwf^{\cxl}_{s}
    \ar[l,"R",bend left]
  \end{tikzcd}
\end{mathpar}
is a Quillen equivalence.
This notion of \emph{Morita equivalence} between type theories has been studied by Isaev~\cite{IsaevMoritaEquivalences}, albeit only for strictly stable type-theoretic structures.

We recall the definition of weak equivalence~\cite{HomotopyTheoryTTs} between CwFs.
\begin{definition}
  Let $F : \CC \to \CD$ be a CwF morphism, where $\CD$ is equipped with weakly stable weak identity types.
  The map $F$ is a \defemph{weak equivalence} if it is essentially surjective on types and terms, \ie{} if it satisfies the following weak type and term lifting conditions:
  \begin{description}
    \item[(weak type lifting)] For every $\Gamma : \Ob_{\CC}$ and type $A : \yo(F(\Gamma)) \to \Ty_{\CD}$, there exists a type $A_{0} : \yo(\Gamma) \to \Ty_{\CC}$ and an equivalence between $F(A_{0})$ and $A$ over $F(\Gamma)$.
    \item[(weak term lifting)] For every $\Gamma : \Ob_{\CC}$, type $A : \yo(\Gamma) \to \Ty_{\CC}$ and term $a : (\gamma : \yo(F(\Gamma))) \to \Tm_{\CD}(F(A)(\gamma))$, there exists a term $a_{0} : (\gamma : \yo(\Gamma)) \to \Tm_{\CC}(A(\gamma))$ and a typal equality between $F(a_{0})$ and $a$ over $F(\Gamma)$.
          \lipicsEnd{}
  \end{description}
\end{definition}

\begin{conjecture}
  The theories of weakly stable weak identity types and strictly stable weak identity types of are Morita equivalent:
  for every $I_{ws}$-cellular model $\CC : \CCwf_{ws}$, the unit $\eta : \CC \to L(\CC)$ is a weak equivalence.
  \lipicsEnd{}
\end{conjecture}

Here the $I_{ws}$-cellular models are the freely generated models in $\CCwf_{ws}$.
Note that they do not coincide with the $I$-cellular CwFs.

We give an informal outline of a likely proof of this result.
We leave a detailed proof to future work.

Fix a $I_{ws}$-cellular model $\CC : \CCwf_{ws}$.
Since $\CC$ is freely generated, it admits a syntactic description and satisfies a universal property; a morphism $\CC \to \CE$ in $\CCwf_{ws}$ is determined by the image of the generating types and terms.

By~\cref{thm:left_strictification_id}, or a generalization to additional type formers, we have a CwF $\CD$, equipped with strictly stable type structures, along with a trivial fibration $F : \CD \to \CC$ in $\CCwf$ that weakly preserves the various type structures.

By induction on the syntax of $\CC$, we construct a morphism $G : \CC \to \CD$ in $\CCwf_{ws}$ along with a homotopy $\alpha : F \circ G \sim \id_{\CC}$.
In other words, we construct a homotopy section $G$ of $F$.
If $F$ was a morphism in $\CCwf_{ws}$, we could obtain a (strict) section from the fact that $\CC$ is cofibrant in $\CCwf_{ws}$ and satisfies a strict lifting property with respect to trivial fibrations.
Since $F$ only preserves the type-theoretic structures weakly, we can only construct a homotopy section.

More precisely, this induction can be described using the \emph{homotopical gluing} of $F : \CD \to \CC$; it is a model $\CG : \CCwf_{ws}$ that classifies the homotopy sections of $F$.
Its objects are triples $(\Gamma,\Delta,e)$, where $\Gamma : \Ob_{\CD}$, $\Delta : \Ob_{\CC}$ and $e$ is an equivalence between $F(\Delta)$ and $\Gamma$.
Its construction ought to be similar to other constructions of homotopical gluing models~\cite{ShulmanUnivalenceInverseDiagrams} and homotopical diagram models~\cite{HomotopicalInverseDiagrams}.

The universal property of $\CC$ then provides a section of $\pi_{2} : \CG \to \CC$, which can be decomposed into a morphism $G : \CC \to \CD$ and a homotopy $\alpha : F \circ G \sim \id_{\CC}$.
\begin{mathpar}
  \begin{tikzcd}
    \CG \ar[r, "\pi_{1}"] \ar[d, "\pi_{2}"] &
    \CD \ar[dl, "F", bend left] \\
    \CC
    \ar[u, dashed, bend left, "\angles{G,\id,\alpha}"]
    &
  \end{tikzcd}
\end{mathpar}

By the universal property of $L(\CC)$, we obtain a map $T : L(\CC) \to \CD$ in $\CCwf_{s}$ such that $T \circ \eta = G$.
\begin{mathpar}
  \begin{tikzcd}
    \CD
    \ar[d, "F"] & \\
    \CC
    \ar[u, "G", bend left]
    \ar[r, "\eta"'] &
    L(\CC)
    \ar[lu, "T"', bend right, dashed]
  \end{tikzcd}
\end{mathpar}

We can now attempt to prove the weak type lifting property for $\eta$.
For any context $\Gamma : \Ob_{\CC}$ and type $A : \yo(\eta \Gamma) \to \Ty_{L(\CC)}$ of $L(\CC)$, we have a candidate lift $F(T(A))[\alpha] : \yo(\Gamma) \to \Ty_{\CC}$.
It remains to prove that $\eta(F(T(A))[\alpha])$ is equivalent to $A$, or equivalently that $\eta(F(T(A)))$ is equivalent to $A$ over the context equivalence $\eta(\alpha_{\Gamma}) : \eta(F(G(\Gamma))) \cong \eta(\Gamma)$.

It suffices to construct a homotopy $\beta : \eta \circ F \circ T \sim \id_{L(\CC)}$ along with a higher homotopy $\gamma$ between the homotopies $\beta \circ \eta$ and $\eta \circ \alpha$.
We expect that these homotopies can be constructed using the universal properties of respectively $L(\CC)$ and $\CC$, by mapping into some other homotopical gluing models.
The weak term lifting property also follows from the existence of these homotopies.

Thus, we have essentially reduced the proof of the Morita equivalence between weakly stable and strictly stable structures to the construction of three homotopical gluing models.

\bibliography{main}

\end{document}